\newif\iftechrep
\techrepfalse

\PassOptionsToPackage{table}{xcolor}
\documentclass[conference, 10pt]{IEEEtran}

\usepackage{macros}
\usepackage{graphicx}
\usepackage[normalem]{ulem}

\usepackage{pifont}
\usepackage{booktabs} 
\usepackage{multirow} 
\usepackage{amsmath}
\usepackage{amssymb,amsfonts}
\usepackage{wasysym}
\usepackage{subfigure}
\usepackage{colortbl}

\PassOptionsToPackage{hyphens}{url}
\usepackage{hyperref}
\usepackage{url}
\makeatletter
\g@addto@macro{\UrlBreaks}{\UrlOrds}
\makeatother

\hypersetup{
	colorlinks=true,
	linkcolor=blue,
	citecolor=red,
	urlcolor=magenta,
}

\usepackage{enumitem}
\usepackage{xspace}

\newcommand{\vincent}[1]{{\color{red}{Vincent: #1}}}
\newcommand{\deepal}[1]{{\color{blue}{Deepal: #1}}}
\newcommand{\cref}[1]{{\S\ref{#1}}}

\newcommand{\remove}[1]{}

\newenvironment{smallitem}{
\begin{itemize}[%
leftmargin=10pt,
labelsep=5pt,
rightmargin=3pt,
labelwidth=0pt,
itemindent=2pt,
listparindent=5pt,
topsep=4pt plus 2pt minus 4pt,
partopsep=4pt,
itemsep=6pt,
parsep=5pt
]
    \setlength{\parskip}{-1pt}
}{\end{itemize}}

\usepackage[ruled,vlined,boxed,linesnumbered,lined,commentsnumbered]{algorithm2e}

\let\emptyset\varnothing

\newtheorem{theorem}{Theorem}
\newtheorem{lemma}{Lemma}
\newtheorem{definition}{Definition}
\newenvironment{proof}{\noindent{\bf Proof.}}{\hfill$\Box${\vspace*{\medskipamount}}}

\usepackage{graphicx}

\usepackage{framed}

\usepackage{tabularx}
\usepackage{listings}
\usepackage{multicol}
\usepackage{color}
\definecolor{black}{RGB}{0,0,0}
\definecolor{gray}{RGB}{102,102,102}        
\definecolor{function}{RGB}{0,102,153}      
\definecolor{lightgreen}{RGB}{102,153,0}    
\definecolor{lightlightgreen}{RGB}{152,193,50}    
\definecolor{bluegreen}{RGB}{51,153,126}    
\definecolor{magenta}{RGB}{217,74,122}  
\definecolor{orange}{RGB}{226,102,26}       
\definecolor{purple}{RGB}{125,71,147}       
\definecolor{green}{RGB}{113,138,98}        
\definecolor{tomato}{RGB}{255,99,71}  
\definecolor{lightred}{RGB}{255,160,131}  
\usepackage[T1]{fontenc}
\usepackage[scaled]{beramono}
\usepackage{microtype}
\lstdefinelanguage{parameterized}{
  firstnumber=1,
  xleftmargin=1.5em,
  numberstyle=\tiny\color{black},
  tabsize=2,
  numbers=left,
  morekeywords = [3]{require,while,if,then,else,do,done,wait,until,end,for,return,returns,upon,from,to,is,in},
  morekeywords = [4]{pragma,function,contract,new,true,false,null,and,or},
  morekeywords = [5]{bench,unlockAccount,createBatch,start_new_consensus,peek,
  propose,execute_transaction,HttpProvider,parse,on,catch,readFileSync,sendSignedTransaction,
  consensus_propose,mvc_propose,break,decide,poll,broadcast,deliver,add,greet,Hello,
  execute,sendTransaction,updateBlockState,executeTx,persist,deserialize,isValid,add},
  morekeywords = [6]{var,public,bytes32,bool,byte,+,=,:=,.,;,,,-,!,=,~,>,<,==,solidity},
  morekeywords = [7]{reliable_broadcast,binary_consensus,consensus},
  keywordstyle = [3]\color{bluegreen},
  keywordstyle = [4]\color{lightgreen},
  keywordstyle = [5]\color{magenta},
  keywordstyle = [6]\color{orange},
  keywordstyle = [7]\color{purple},
  sensitive = true,
  morecomment = [l][\color{gray}]{//},
  morecomment = [s][\color{gray}]{/*}{*/},
  morecomment = [s][\color{gray}]{/**}{*/},
  morestring = [b][\color{purple}]",
  morestring = [b][\color{purple}]',
  literate=
  {á}{{\'a}}1 {é}{{\'e}}1 {í}{{\'i}}1 {ó}{{\'o}}1 {ú}{{\'u}}1
  {Á}{{\'A}}1 {É}{{\'E}}1 {Í}{{\'I}}1 {Ó}{{\'O}}1 {Ú}{{\'U}}1
  {à}{{\`a}}1 {è}{{\`e}}1 {ì}{{\`i}}1 {ò}{{\`o}}1 {ù}{{\`u}}1
  {À}{{\`A}}1 {È}{{\'E}}1 {Ì}{{\`I}}1 {Ò}{{\`O}}1 {Ù}{{\`U}}1
  {ä}{{\"a}}1 {ë}{{\"e}}1 {ï}{{\"i}}1 {ö}{{\"o}}1 {ü}{{\"u}}1
  {Ä}{{\"A}}1 {Ë}{{\"E}}1 {Ï}{{\"I}}1 {Ö}{{\"O}}1 {Ü}{{\"U}}1
  {â}{{\^a}}1 {ê}{{\^e}}1 {î}{{\^i}}1 {ô}{{\^o}}1 {û}{{\^u}}1
  {Â}{{\^A}}1 {Ê}{{\^E}}1 {Î}{{\^I}}1 {Ô}{{\^O}}1 {Û}{{\^U}}1
  {Ã}{{\~A}}1 {ã}{{\~a}}1 {Õ}{{\~O}}1 {õ}{{\~o}}1
  {œ}{{\oe}}1 {Œ}{{\OE}}1 {æ}{{\ae}}1 {Æ}{{\AE}}1 {ß}{{\ss}}1
  {?}{{\H{u}}}1 {?}{{\H{U}}}1 {?}{{\H{o}}}1 {?}{{\H{O}}}1
  {ç}{{\c c}}1 {Ç}{{\c C}}1 {ø}{{\o}}1 {å}{{\r a}}1 {Å}{{\r A}}1
  {€}{{\euro}}1 {£}{{\pounds}}1 {«}{{\guillemotleft}}1
  {»}{{\guillemotright}}1 {ñ}{{\~n}}1 {Ñ}{{\~N}}1 {¿}{{?`}}1
}
\lstdefinelanguage{Solidity}{
        tabsize=2,
        numbers=left,
        stepnumber=1,
        showstringspaces=false,
	keywords=[1]{anonymous, assembly, assert, balance, break, call, callcode, case, catch, class, constant, continue, constructor, contract, debugger, default, delegatecall, delete, do, else, emit, event, experimental, export, external, false, finally, for, function, gas, if, implements, import, in, indexed, instanceof, interface, internal, is, length, library, log0, log1, log2, log3, log4, memory, modifier, new, payable, pragma, private, protected, public, pure, push, require, return, returns, revert, selfdestruct, send, solidity, storage, struct, suicide, super, switch, then, this, throw, transfer, true, try, typeof, using, value, view, while, with, addmod, ecrecover, keccak256, mulmod, ripemd160, sha256, sha3}, 
	keywordstyle=[1]\color{lightgreen}\bfseries,
	keywords=[2]{address, bool, byte, bytes, bytes1, bytes2, bytes3, bytes4, bytes5, bytes6, bytes7, bytes8, bytes9, bytes10, bytes11, bytes12, bytes13, bytes14, bytes15, bytes16, bytes17, bytes18, bytes19, bytes20, bytes21, bytes22, bytes23, bytes24, bytes25, bytes26, bytes27, bytes28, bytes29, bytes30, bytes31, bytes32, enum, int, int8, int16, int24, int32, int40, int48, int56, int64, int72, int80, int88, int96, int104, int112, int120, int128, int136, int144, int152, int160, int168, int176, int184, int192, int200, int208, int216, int224, int232, int240, int248, int256, mapping, string, uint, uint8, uint16, uint24, uint32, uint40, uint48, uint56, uint64, uint72, uint80, uint88, uint96, uint104, uint112, uint120, uint128, uint136, uint144, uint152, uint160, uint168, uint176, uint184, uint192, uint200, uint208, uint216, uint224, uint232, uint240, uint248, uint256, var, void, ether, finney, szabo, wei, days, hours, minutes, seconds, weeks, years},	
	keywordstyle=[2]\color{bluegreen}\bfseries,
	keywords=[3]{block, blockhash, coinbase, difficulty, gaslimit, number, timestamp, msg, data, gas, sender, sig, value, now, tx, gasprice, origin},	
	keywordstyle=[3]\color{function}\bfseries,
	identifierstyle=\color{black},
	sensitive=false,
	comment=[l]{//},
	morecomment=[s]{/*}{*/},
	commentstyle=\color{gray}\ttfamily,
	stringstyle=\color{purple}\ttfamily,
	morestring=[b]',
	morestring=[b]"
}
\newcommand\Small{\fontsize{8}{8.8}\selectfont}
\newcommand*\LSTfont{\Small\ttfamily\SetTracking{encoding=*}{-60}\lsstyle}
\newcommand\TextSize{\fontsize{8.5}{9.5}\selectfont}
\newcommand*\ttt{\TextSize\ttfamily\SetTracking{encoding=*}{-60}\lsstyle}

\makeatletter
\newcommand\footnoteref[1]{\protected@xdef\@thefnmark{\ref{#1}}\@footnotemark}
\makeatother
\begin{document}
\date{}

\newcommand{\solution}{CollaChain\xspace}
\newcommand{\Middleware}{Middleware\xspace}
\newcommand{\middleware}{middleware\xspace}
\newcommand{\consensus}{DBFT\xspace}

\title{\Large \bf \solution: A BFT Collaborative \Middleware for Decentralized Applications
} 

\author{Deepal Tennakoon\\ dten6395@uni.sydney.edu.au \and Yiding Hua\\ yhua7740@uni.sydney.edu.au \and Vincent Gramoli\\ vincent.gramoli@sydney.edu.au}

\maketitle
\thispagestyle{plain}
\pagestyle{plain}

\begin{abstract}
The sharing economy is centralizing services, leading
to misuses of the Internet: we can list growing damages of data hacks, global outages, and even uses of data to manipulate their owners. Unfortunately, there is no decentralized web where users can interact peer-to-peer in a secure way. Blockchains incentivize participants to individually validate every transaction and impose
their block to the network. As a result, the validation of smart contract requests is computationally intensive while the agreement on a unique state does not make full use of the network.

In this paper, we propose CollaChain, a new byzantine
fault tolerant blockchain compatible with the largest ecosystem of DApps that leverages collaboration. First, the participants executing smart contracts collaborate to validate the transactions, hence halving the number of validations required by modern blockchains (e.g., Ethereum, Libra). Second, the participants
in the consensus collaborate to combine their block proposal into a superblock, hence improving throughput as the system grows to hundreds of nodes. In addition, CollaChain offers the possibility to its users to interact securely with each other without downloading the blockchain, hence allowing interactions
via mobile devices. CollaChain is effective at outperforming the Concord and Quorum blockchains and its throughput peaks at 4500 TPS under a Twitter DApp (Decentralized Application) workload. Finally, we demonstrate CollaChain’s scalability by deploying it on 200 nodes located in 10 countries over 5 continents.
\end{abstract}

\section{Introduction}
As the number of misuses of Internet data grows, 
so does the need for 
decentralized \middleware rewarding individuals for sharing data. 
These misuses stem from a ``centralization'' of the web according to Turing-awardee Tim Berners-Lee, who contributed to the design of a decentralized alternative to his 30-year-old Web~\cite{MSH16}.
At its origin, the Web helped users communicate with each other through their desktop computers. In the 2000s users started sharing data through Google, Facebook, Microsoft, and Amazon. By 2025, the sharing economy---to which users can contribute with their mobile phone---is expected to represent \$335 billion~\cite{pwc15}. This centralization has severe drawbacks: it exposes data to leaks and hacks~\cite{Fac18} and
it facilitates user manipulation~\cite{Pra18}.

Decentralized applications, or \emph{DApps} for short, are increasingly popular at allowing users to trade services peer-to-peer without transferring ownership.
In the third quarter of 2020, DApps transaction volume experienced an 8-fold increase reaching \$125B~\cite{Dap20}. With 96\% of this volume occurring on top of the Ethereum blockchain~\cite{wood2014ethereum} alone, most DApps are, however, plagued by Ethereum capacity capped at $\sim$15 transactions per second (TPS)~\cite{CHEN2019101055}.
Even without forcing miners to resolve a crypto-puzzle to obtain a proof-of-work, the proof-of-authority alternative of Ethereum delivers $\sim$80\,TPS~\cite{LCG20}.
Ethereum 
suffers from (i)~an expensive validation during the execution of smart contracts as well as (ii)~an incentive for consensus participants to compete into a fierce battle that aims at imposing their block to the rest of the system.
As a result, Ethereum already experienced congestion due to DApps~\cite{Cry17} and its capacity is inherently 
too low to support a DApp with the 4000+\,TPS throughput of Twitter~\cite{LSM11}.

In this paper, we propose \solution,
a new collaborative byzantine fault tolerant \middleware designed for decentralizing the sharing economy.
It is compatible with the largest ecosystem of DApps as it runs an adapted version of the Ethereum Virtual Machine (EVM), called Scalable EVM (SEVM), and decouples the traditional blockchain design into an (i)~SEVM component that divides the validation of EVM nodes by two as the number $n$ of blockchain participants tends to infinity and (ii)~a distributed consensus component that 
 potentially decides a number of transactions proportional to $n$,
hence addressing the two aforementioned limitations to scale as $n$ increases. 
Although a SEVM node also validates transactions eagerly to mitigate denial-of-service (DoS) attacks, we limit the number of nodes 
receiving this transaction to achieve (i).
Consensus participants collaborate to combine
their proposal into the same \emph{superblock} decision,
as opposed to the competitive classic consensus approach, to achieve (ii). 

Perhaps more importantly, a user of \solution, who wants to ensure its DApp can access a consistent state of the blockchain without trusting a central entity, does not need to download the blockchain.
\solution simply requires a web-based or mobile app
and leverages the upper-bound $f<n/3$ on the number of arbitrary (byzantine) failures among $n$ servers to reach consensus~\cite{LSP82}.
Other blockchains typically require users to either download block headers to verify that the current state is consistent, a process called ``synchronizing'', before issuing requests~\cite{eth-docs-clients} or trust a central entity, which defeats the blockchain purpose.
Ethereum fast synchronization requires more than 280\,GB of free storage space and takes 4 hours on average 
on an i3.2xlarge AWS EC2 instance with 8 vCPUs, 61\,GiB memory and 1.9\,TiB NVMe SSD~\cite{eth-fast-sync}, a task nearly impossible for any mobile device.
One may think of downloading less information with a ``light'' synchronization, however, the corresponding validation 
cannot guarantee that the blockchain is correct
due to incomplete blockchain records~\cite{eth-sync-modes}. 

\solution achieves 2K\,TPS when deployed on 200 machines spread in 10 countries over five continents. 
As indicated in Table~\ref{table:comparison}, it outperforms the non-sharded blockchains that tolerate byzantine failures. Note that we discuss sharding as an orthogonal optimization in \cref{sec:rw}.
\solution's throughput is two orders of magnitude larger than the Ethereum capacity throughput and, although not reported in Table~\ref{table:comparison}, 
one order of magnitude faster than the 172\,TPS of the recent SBFT in a world-scale setting~\cite{GAG19}. 
To the best of our knowledge, the non-sharded blockchains that outperform \solution are the ones that only tolerate crash failures.
We compare \solution performance to Quorum and Concord 
(\cref{sec:comparison}), show the impact of validation reduction on performance using BlockBench~\cite{DWC17}, illustrate the scalability of \solution 
to hundreds of machines over 5 continents, and demonstrate a 4500\,TPS peak throughput when running the Twitter DApp of the {\sc Diablo} benchmarking framework~\cite{BGG21}.

\begin{table}[ht]
\setlength\tabcolsep{8pt}
\begin{tabular}{rlcc}
\toprule
Blockchain & Fault tolerance & Smart contract & TPS \\ 
\midrule
Ethereum v1.x~\cite{wood2014ethereum} & \cellcolor{lightred}probabilistic &  \cellcolor{lightlightgreen}Solidity & \cellcolor{tomato}15 \\ 
Avalanche~\cite{TR18} & \cellcolor{lightred}probabilistic & \cellcolor{lightred}C-Chain & \cellcolor{lightred}1.3K \\ 
Algorand~\cite{GHM17} &  \cellcolor{lightred}probabilistic &  \cellcolor{lightred}{\sc teal} & \cellcolor{lightred}1K \\ 
Hyperledger Fabric~\cite{ABB18}  & \cellcolor{tomato}crash &  \cellcolor{lightred}ChainCode  & \cellcolor{lightlightgreen}3K \\ 
FastFabric~\cite{GLG19}  & \cellcolor{tomato}crash & \cellcolor{lightred}ChainCode & \cellcolor{lightlightgreen}20K \\ 
Cosmos/Ethermint~\cite{CFM21} & \cellcolor{lightlightgreen}byzantine & \cellcolor{lightlightgreen}Solidity & \cellcolor{tomato}438 \\ 
Burrow~\cite{burrow} & \cellcolor{lightlightgreen}byzantine & \cellcolor{lightlightgreen}Solidity & \cellcolor{tomato}765 \\ 
SBFT~\cite{GAG19}  & \cellcolor{lightlightgreen}byzantine & \cellcolor{lightlightgreen}Solidity &  \cellcolor{tomato}378 \\ 
Stellar~\cite{LLM19} & \cellcolor{lightlightgreen}byzantine & \cellcolor{lightred}SSC & \cellcolor{tomato}100 \\
\solution & \cellcolor{lightlightgreen}byzantine  & \cellcolor{lightlightgreen}Solidity & \cellcolor{lightlightgreen}4.5K \\ 
\bottomrule
\end{tabular}
    \vspace{0.5em}
\caption{Comparison of (non-sharded) blockchains with potential DApp support. 
The performance of Ethereum was taken from~\cite{CHEN2019101055}, the performance of Hyperledger Fabric and FastFabric was taken from~\cite{GLG19} (their BFT orderer~\cite{SBV18} is not listed as it does not
make these whole blockchains tolerate byzantine failures). 
Peak smart contract performance of SBFT was measured in a continent-scale setting; it actually achieves 172\,TPS in a world scale setting~\cite{GAG19}.
The throughput of Burrow is taken from~\cite{SNG20}.
The performance of Avalanche was taken from~\cite{TR18}, even if recent claims announced higher values obtained in unknown conditions.
"Solidity" indicates that the corresponding blockchain is compatible with the large ecosystem of Ethereum DApps.
We discuss the performance of sharded blockchains in Section~\cref{sec:sharding}. 
\label{table:comparison} 
}
\end{table}

\solution builds upon various results. 
The superblock optimization already appeared in a UTXO-based blockchain~\cite{CNG21}, however, \solution applies it to 
smart contracts by decoupling their execution and persistent storage into sub-tasks to avoid request losses as we will illustrate in \cref{ssec:store}.
%
Its SEVM nodes receive the requests from the clients and batch them into blocks that are sent to consensus nodes. 
Similar to an Ethereum~\cite{wood2014ethereum,But20} server (i.e., miner) or a Libra~\cite{BBC19} server (i.e., validator), an SEVM server of \solution  validates eagerly a transaction upon reception and  validates lazily the same transaction upon execution
(after the block that contains it is agreed upon).
In contrast with these blockchains, an SEVM node does not propagate the transactions to other nodes upon reception from the client, hence reducing the number of eager validations.
Similar to byzantine fault tolerant (BFT) replicated state machine protocols~\cite{BSA14,YMR19}, the consensus nodes decide on a unique batch of transactions without assuming synchrony as long as less than a third of consensus nodes are byzantine, which is resilient optimal~\cite{LSP82}.

In the remainder of the paper, we present our motivations and the necessary background (\cref{sec:bgrd}), as well as our goals and assumptions (\cref{sec:goal}).
We then present \solution (\cref{sec:solution}) 
that we prove correct, and evaluate it in a geodistributed setting 
and compare it against other blockchains (\cref{sec:evaluation}).
Finally, we present the related work (\cref{sec:rw}) and conclude (\cref{sec:conclusion}).
A smart contract to reconfigure the nodes of \solution is provided in Appendix~\ref{appendix}.

\section{Background and Motivations}\label{sec:bgrd}

\paragraph{Decentralized applications rationale}
Decentralized applications (DApps) alleviate many problems induced by the centralization of the sharing economy.
To mention a few, YouTube exprienced an outage~\cite{Bria20} that DTube could have remedied by sharing videos peer-to-peer~\cite{dtube19}.
Uber drivers feel manipulated by an opaque matching algorithm~\cite{MZ17}, whereas the DApp counterpart, called Drife, could offer transparency~\cite{drife}.
So what is the performance necessary to implement such a decentralized version of the sharing economy?
To answer this question, let us consider Twitter, which is a popular micro-blogging application.
Twitter experiences more than 4000\,tweets per second on average and its peak demand largely exceeds this number~\cite{LSM11}.
It is thus crucial for a mainstream decentralized \middleware to support thousands of transactions per second. Unfortunately most blockchains cannot (cf. Table~\ref{table:comparison}).


\paragraph{The redundant validations of Ethereum}
Ethereum~\cite{wood2014ethereum} features the Ethereum Virtual Machine (EVM) that was proposed in part to cope with the limited expressiveness of Bitcoin~\cite{Nak08} and to execute DApps written in a Turing complete programming language as \ms{smart contracts}. 
Go Ethereum, or {\ttt geth} for short, is the mostly deployed Ethereum implementation~\cite{gethstats}.
In order to check that a request (or transaction)
is valid, all of the {\ttt geth} \emph{servers} (i.e., miners) must validate twice each executed transaction:
\begin{itemize}
\item {\bf Eager validation:} This validation occurs 
upon reception of a new client transaction and checks the nonce value; that the sender account has sufficient balance; that the gas is sufficient to execute the transaction; and the transaction does not exceed the block gas limit, is signed properly and is not oversized. It reduces 
the effect of denial-of-service (DoS) attacks as an invalid transaction is dropped early. If the transaction is valid, it is propagated to other servers.
\item {\bf Lazy validation:} This validation occurs before transactions are executed in a decided block and simply checks the nonce and whether there is enough gas for execution.
This lazy validation is necessary to guarantee that transactions in a newly received decided block are indeed valid. 
The lazy validation is thus less time consuming in {\ttt geth} than the eager validation, 
this is why we focus on reducing the number of eager validations.
\end{itemize}

This is an overconservative strategy
because each to-be-executed transaction of {\ttt geth} is 
validated twice by each server. This is unnecessary as an invalid transaction coming from a byzantine node will either be dropped by lazy validation prior to execution or fail execution and the state reversed if there is an invalidity not checked by the lazy validation.

It is interesting to note, also, that in a system where few replicas are byzantine, there is no 
need for all servers to validate all transactions twice.
%
We explain in~\cref{sec:sevm} how, without reducing security, we reduce the number $k$ of eager validations per server down to $k/n$ to scale to a large system size $n$.

\paragraph{The inefficiency of byzantine fault tolerant consensus}
For security reasons, a blockchain must guarantee that nodes agree on a unique block at each index of the chain.
To cope with malicious (or \emph{byzantine}) participants, this requires solving the byzantine fault tolerant (BFT) consensus problem~\cite{PSL80}, where
every non-byzantine or \emph{correct} node eventually decides a value such that no two correct nodes decide differently.
Unfortunately, traditional consensus protocols solve this problem by electing a leader node that tries to impose its value to the other nodes~\cite{CL02,BSA14,BKM18,YMR19}.
While these leader-based designs proved effective in local area networks to deploy a secure version of the Network File System~\cite{CL02}, it 
generally cannot scale to 
large blockchain networks because only one value is decided~\cite{Buc16,VG21}, regardless of the number of proposed values in the system. 
Recent consensus implementations allowed to commit up to as many proposed values per consensus instance as participating nodes to scale performance~\cite{MSC16,KBBR21,CGLR18}.
However, some of these variants~\cite{MSC16} fail at solving consensus because their binary consensus protocol may not terminate~\cite{TG19}.
And the only blockchains that integrate a provably correct consensus implementation that combines block proposals into a \emph{superblock}
support simple transactions but cannot execute arbitrary programs or smart contracts~\cite{KBBR21,CNG21}. 
In~\cref{sec:sevm}, we will describe the obstacles we overcome to combine multiple proposals of smart contracts  creations and invocations 
during each consensus execution.
%
%
%

\section{Goals and Assumptions}\label{sec:goal}

We consider an \emph{open permissioned} blockchain model~\cite{MSC16,CNG21} in that a subset of the distributed machines have the permission to run 
the current instance of the consensus, or to execute smart contracts and transaction requests as well as to maintain the resulting state. 
This model is called ``open'' as permission can be revoked and we do not prevent a particular node from obtaining a permission later on: as opposed to Ethereum 
we simply prevent all nodes from providing the same service at the same time to avoid resource waste (\cref{ssec:open}).

We assume \emph{partially synchronous} communication and computation in that the upper bound on the time it takes for a step exists but is unknown~\cite{DLS88}.
For simplicity, we assume that each permissioned participant 
runs both a consensus node and a state node and that up to $f$ of these participants (and any of their nodes) can fail arbitrarily by being byzantine. In this case, we call such a participant as a \emph{blockchain node}.

\paragraph{The Blockchain problem}\label{ssec:bc}
We refer to the blockchain problem as the problem of ensuring both the
safety and liveness properties that were defined in the literature by Garay et al.~\cite{GKL15} and restated more recently by Chan et al.~\cite{CS20}, and a classic validity property~\cite{CNG21}.
\begin{definition}[The Blockchain Problem]\label{def:blockchain}
The \emph{blockchain problem} is to ensure that a distributed set of blockchain nodes 
maintain a sequence of transaction blocks such that the three following properties hold:
\begin{itemize}
\item  \emph{Liveness:} if a correct blockchain node receives a transaction, then this transaction will eventually be reliably stored in the block sequence of all correct blockchain nodes.
\item \emph{Safety:} the two chains of blocks maintained locally by two correct blockchain nodes are either identical or one is a prefix of the other. 
\item \emph{Validity:} each block appended to the blockchain of each correct blockchain node is a set of valid transactions (non-conflicting well-formed transactions that are correctly signed by its issuer).
\end{itemize}
\end{definition}
The safety property does not require correct blockchain nodes to share the same copy, simply because one replica may already have received the latest block before another receives it.
Note that, as in classic definitions~\cite{GKL15,CS20}, the liveness property does not guarantee that a client transaction is included in the blockchain: if a client sends its transaction request exclusively to byzantine nodes then byzantine nodes may decide to ignore it.

\paragraph{Our goal of a secure and efficient middleware for DApps}
Our goal is thus to support DApps, by allowing clients (i)~to access consistent data despite $f<n/3$ byzantine servers through all sorts of devices and (ii)~to serve a large demand generated by network effects
as follows:
\begin{enumerate}
\item {\bf Lightweight-security:} the users should be able to securely interact with the blockchain from various devices.
To access apps, users typically use handheld devices that cannot download blockchain histories due to resource constraints (Ethereum history exceeds 280\,GiB~\cite{eth-fast-sync}).
yet they need to interact securely despite unpredictable message delays.
\item {\bf Thousands-TPS:} the volume of transactions that can be served per second should prevent a backlog of requests that grows and leads to congestion. We know that DApps create congestion on Ethereum~\cite{Cry17} and EOS~\cite{edois}, and popular applications, like Twitter, exceed 4000 requests per second~\cite{LSM11}.
\end{enumerate}
Property (1) alleviates the need for clients to download the blockchain history, but requires them to interact securely, which is enabled by limiting the number of failures to $f$ and querying $f+1$ identical copies of the current state to retrieve the correct information as detailed in~\cref{sec:solution}.
Also, we know that to allow users to issue (potentially conflicting) transactions from distinct devices, we need to solve consensus~\cite{GKM19}.
Property (2) lower bounds the capacity to around 2000\,TPS to serve the demand of sharing applications. 
Although this might be insufficient to run multiple DApps, we explain in \cref{sec:sharding} how to shard \solution and deploy different DApps to different shards.
This would help \solution support many DApps smoothly, given its 4500\,TPS peak throughput (\cref{sec:dapp}).

\section{\solution}\label{sec:solution}

\solution is a collaborative blockchain compatible with the largest ecosystem of DApps, it is optimally resilient against byzantine failures. 
The layered architecture is depicted in Fig.~\ref{fig:architecture} with a Scalable EVM (SEVM) node at the top and a consensus node at the bottom that can be run on the same machine as a single \emph{blockchain node}.
The communication between the consensus node and the SEVM node is event-based and implemented with {\ttt gRPC}.  
Although this presents an execution overhead as both the consensus node and the SEVM node can typically execute on a single machine, this offers greater modularity.

\begin{figure}[t]
\begin{center}
\includegraphics[scale=0.59]{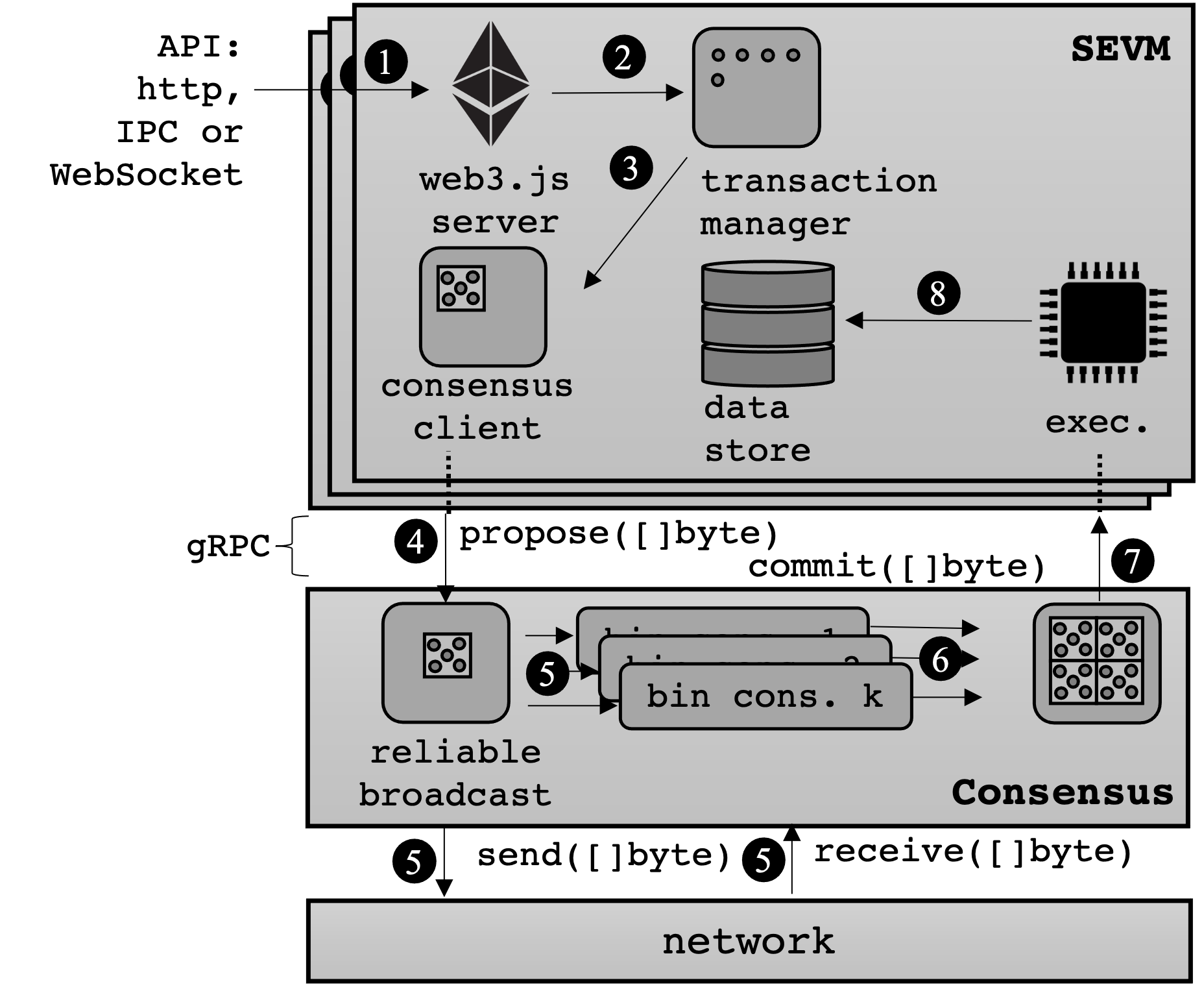}
\caption{The architecture of \solution. \ding{202} A client sends a transaction to some replica(s), 
\ding{203} 
at each replica the web3.js server validates transactions and sends them to the transaction manager that 
\ding{204} sends a block to the consensus client. \ding{205} The consensus client $\lit{propose}$s it to the consensus protocol. Upon reception of a new block from the consensus client, \ding{206} the consensus protocol sends it through the network with a reliable broadcast. Remote replicas start participating in the same instance (if not done yet) upon reliably delivering this proposed block. \ding{207} When the consensus outputs some acceptable blocks, \ding{208} all of these blocks are combined into a superblock and sent to the SEVM. \ding{209} 
As in the EVM, the SEVM is responsible for executing and storing blocks, except that the SEVM will store multiple blocks per consensus instance.
\label{fig:architecture}
}
\end{center}
\end{figure}

\solution offers a secure interface to lightweight clients (\cref{ssec:interface}) and to scale to a geodistributed network of $n$ blockchain nodes, \solution
reduces the time spent validating transactions as $n$ grows (\cref{sec:sevm}) and increases the amount of blocks committed per consensus instance as $n$ grows (\cref{sec:dbft}). Before discussing these, we present the \solution overview through the transaction lifecycle (\cref{sec:lifecycle}).

%

\subsection{The transaction lifecyle}\label{sec:lifecycle}

In the following, we use the term \emph{transaction} to indistinguishably refer to a simple asset transfer, the upload of a smart contract or the invocation of a smart contract function. 
The lifecycle of a transaction goes through these subsequent stages: 
\begin{itemize}
\setlength{\itemindent}{.56in}
\item[{\bf 1. Reception.}] The client creates a properly signed transaction and sends it to at least one \solution  node. Once a request containing the signed transaction is received \ding{202} by the JSON RPC server of the SEVM state machine running within \solution, the eager validation (\cref{sec:bgrd}) starts. 
If the validation fails, the transaction is discarded.
If the validation succeeds, the transaction is added to a transaction pool. Unlike in Ethereum where the transaction would be propagated 
to all miners increasing the number of eager validations, \solution simply proposes it to the consensus node as follows.
If the number of transactions in the pool reaches a threshold, 
then the transaction manager creates a new proposed block
with a number (defined by the threshold) of transactions from the pool \ding{203}.
It serializes and sends the proposed block to the consensus client \ding{204}. 
For the sake of our superblock optimization (cf.\cref{ssec:superblock}) and in contrast with Ethereum, the proposed block does not contain a hash just yet.
\item[{\bf 2. Consensus.}] Once the consensus client receives a proposed block, it sends the corresponding byte array to the consensus system by invoking the {\ttt propose([]byte)} method \ding{205}. The consensus system starts a new instance of consensus using the new block if it is not currently part of another consensus instance. Otherwise, it adds the new block to the block queue waiting for the current consensus instance to terminate. Like in classic reductions of the general consensus problem to the binary consensus problem~\cite{BCG93,BKR94}, \solution's consensus execution \ding{206} consists of  an all-to-all reliable broadcast of the blocks among all consensus replicas that trigger as many \emph{binary} consensus instances whose outputs indicate the indices of acceptable blocks \ding{207} as detailed in~\cref{sec:dbft}. The consensus system creates a superblock with all acceptable blocks (\cref{ssec:superblock}) and sends this superblock to the state machine by invoking the {\ttt commit([]byte)} method \ding{208}.
\setlength{\itemindent}{.43in}
\item[{\bf 3. Commit.}] When the superblock is received by the {\ttt gRPC} server running in the SEVM state machine, the superblock is first deserialized using JSON unmarshalling. The SEVM does the lazy validation (\cref{sec:bgrd}) before committing the deserialized transactions \ding{209}. Note that as opposed to the eager validation, all servers execute the lazy validation for a committed transaction, yet it does not prevent \solution from scaling to hundreds of nodes (\cref{sec:geodistributed}). Once the superblock is decided, each of its blocks are executed, their hash is included, their results are written to persistent storage on the local disk and the lifecycle ends.
\end{itemize}

\subsection{Secure interface for lightweight clients}\label{ssec:interface}
As opposed to classic blockchains, like Ethereum, \solution does not require the client\footnote{The term ``client'' is often used in Ethereum to refer to a node regardless of whether it acts as a server. We use client in the traditional sense of the client-server distinction~\cite{TvS07}.} interacting with the service to download the blockchain or its block headers.
Instead, \solution accepts connections from simple javascript-enabled browsers, as can be found on mobile devices~\cite{metamask}. Similar to {\ttt geth}, \solution supports a {\ttt web3.js} API that allows the user to communicate through {\ttt http}, IPC or websocket. 
Given that \solution tolerates $f$ failures, it is sufficient for the client to query the same copy of the world state from $f+1$ distinct blockchain servers to guarantee that this copy is consistent. And the client is guaranteed to 
find this copy at $f+1$ blockchain nodes by contacting $2f+1$ blockchain nodes by assumption. 
As a result, the client interacts securely with \solution without any blockchain records whereas an Ethereum ``light'' client cannot interact securely due to incomplete blockchain records~\cite{eth-sync-modes}. 
Hence \solution guarantees the \emph{Lightweight-security} property (\cref{sec:goal}).

\subsection{From the EVM to SEVM}\label{sec:sevm}

Here we present the modifications we made to the original EVM (and in particular {\ttt geth} v1.8.27) in order to obtain the \emph{Scalable EVM}, or \emph{SEVM} for short.
More specifically, provided that $k$ transactions are received by \solution, we reduce the average number $k$ of transactions each SEVM node eagerly validates 
to $k/n$.

\subsubsection{Reducing the transaction validations}\label{ssec:tx-validation}
As opposed to each Ethereum server that validates 
eagerly and lazily each of the $k$ transactions of the system, 
each of the $n$ \solution servers eagerly validates on average $k/n$ transactions.
Specifically, only one SEVM node  needs to eagerly validates each transaction: the first SEVM node receiving the transaction validates it but does not propagate it to other SEVM nodes but simply proposes it to the consensus. 
As a result, \solution limits the redundant validations, which improves performance. More precisely, if the number of SEVM nodes is $n$, then each SEVM node does $1+1/n$ validations per transaction on average (one lazy validation + $1/n$ eager validation) compared to the two validations needed in {\ttt geth}. As $n$ tends to infinity, \solution servers validate on average half what {\ttt geth} servers validate.
In the worst case, where all clients send their transactions to $f+1 = n/3$ servers simultaneously, then each server will still eagerly validate only $k/3$ transactions.

Note that, as a result of our optimization, a byzantine SEVM node could propose transactions to the consensus without validating them eagerly, in this case two things can happen: i) The transaction is discarded at the lazy validation if invalid ii) The SEVM attempts to execute the invalid transaction, fails at execution and reverses the state to what it was. Either way, there is no impact on the safety of the blockchain. This is also not a DoS vulnerability of \solution, as even a byzantine EVM node in Ethereum can propagate invalid transactions to all EVM nodes, forcing all EVM nodes to unnecessarily eager-validate them.

Finally, reducing the validation needed at each SEVM node helps \solution reach the \emph{Thousands-TPS} property (\cref{sec:goal}).

\subsubsection{Reliably storing superblocks}\label{ssec:storing}
At each index of the blockchain, our SEVM typically executes many more transactions as part of function {\ttt execute\_transaction} (lines~\ref{line:exec-tx-start}--\ref{line:exec-tx-end}) than Ethereum.
This is due to the consensus outputting through the {\ttt commitChan} channel a superblock containing potentially as many blocks as blockchain nodes (line~\ref{line:commit-chan}). 
In Ethereum, blocks are created before the consensus, thus {\ttt geth} updates only one block, by setting
its state parameters,
per consensus instance:
{\ttt updateBlockState} points a block to its parent, assigns the block header timestamp
and the number of transactions associated with a block. 
This function should thus be invoked for each block before the transactions of the block are executed and persisted, in order to ensure that the data structures are updated properly.
To store multiple blocks at the end of the consensus instance, 
we modified {\ttt geth} to 
{\ttt updateBlockState} (line~\ref{line:update-block-state}) multiple times 
per consensus instance (one invocation per block) as follows:

More specifically, we reordered the transactions (as disordered transactions could be discarded due to their invalid nonces) and changed the original procedure 
to guarantee that not only one block but all blocks of our superblock were correctly stored in the transaction and reception tries as a batch of $n$ blocks. 
Like the C++, python and {\ttt geth} software of Ethereum, 
we reliably store the information in the open source key-value store LevelDB (line~\ref{line:persist}).

\begin{lstlisting}[language=parameterized,basicstyle=\LSTfont,escapechar = ?,escapeinside={(*}{*)},frame = single]
execute_transaction: (*\label{line:exec-tx-start}*)
  // for each superblock received 
  for superblock in node.commitChan do (*\label{line:commit-chan}*)
    vtxs := (*$\emptyset$*) // set of valid transactions
    for block in superblock do // each block (*\label{line:block-loop-start}*)
      txs := node.txm.deserialize(block) // get txs
      for tx in txs:
        if isValid(tx): vtxs.add(tx) // lazy validation (*\label{line:validate}*)
      // set the corresponding block state and order txs
      updateBlockState() (*\label{line:update-block-state}*)
      for tx in vtxs do // for each valid tx...
        executeTx(tx) // ...execute it
      done
      persist(vtxs) // persist valid txs to disk (*\label{line:persist}*)
    done  (*\label{line:block-loop-end}*)
  done (*\label{line:exec-tx-end}*)
\end{lstlisting}

\subsubsection{SEVM support for fast-paced consecutive blocks}
Since our consensus system is fast, it creates and delivers superblocks at high frequency through the commit channel to the SEVM. As {\ttt geth} does not expect to receive blocks at such a high frequency, it raises an exception outlining that consecutive block timestamps are identical, which never happens in a normal execution of Ethereum.
This equality arose because  {\ttt geth} encodes the timestamp of each block as {\ttt uint64}, not leaving enough space for encoding time with sufficient precision. {\ttt geth} typically reports an error when consecutive timestamps are identical, due to a strict check that compares the parent block timestamp to the current block timestamp in {\ttt go-ethereum/consensus/ethash/consensus.go}: {\ttt header.Time < parent.Time}.
We changed the original check to {\ttt header.Time <= parent}, which allowed for fast-paced executions of consecutive blocks.

\subsubsection{Bypassing the SEVM resource bottlenecks}
After the consensus, the SEVM lazily validates many transactions, updates the memory 
and storage, 
which consumes high CPU, memory and IO resources. 
Typically, high CPU usage slows down the SEVM which results in the increase of the pending list of transactions. Once a threshold of pending transactions are reached, we observe transaction drops. 
This was evident in our superblock implementation. We observed that consuming each resource one after another, for 10 proposed blocks with a total of 15,000 transactions, would 
lead to losing transactions requests 
even on our reasonably-provisioned AWS instances featuring 16\,GB RAM and 4 vCPUs (Fig.~\ref{fig:superblock}).
This is why we made SEVM fully process one proposed block of the superblock at a time
allowing it to alternate frequently between CPU-intensive (verifying signatures and transaction executions) and memory-intensive (state trie write) and IO-intensive (reception/transaction tries writes) tasks.
Thanks to this optimized implementation of the superblock, SEVM does not experience bottlenecks as the number of nodes increases (cf.~\cref{sec:geodistributed}).

\subsection{A BFT Consensus for SEVM}\label{sec:dbft}

As opposed to smart contract blockchains that decide (at most) one of the proposed blocks,
\solution decides a superblock that results from its consensus system combining multiple proposed blocks into a single decision.
In the ideal case, agreeing on a superblock thus allows to commit $\Omega(n)$ blocks of distinct transactions at the end of a single consensus instance.

\subsubsection{Peer-to-peer network}

The peer-to-peer (P2P) network of the consensus system is implemented using golang's RPC package {\ttt net/rpc}. The consensus node reads consensus network configuration from a {\ttt yaml} configuration file upon initialization. The configuration file contains the network size $n$, a port number $p$ and a list of socket addresses specified in {\ttt ip:port} format. The consensus node sets up a gRPC server on port $p$ for consensus messages. The list of socket addresses are gRPC endpoints of consensus nodes. To prevent byzantine nodes from eavesdropping, all communications use TLS. We show that the overhead induced by the encryption layer of TLS is negligible in Fig.\ref{fig:geo} of~\cref{sec:evaluation}.
 


\subsubsection{Increasing the decision size}\label{ssec:superblock}

The requirement of deciding at most one block is too restrictive to scale with the number $n$ of consensus participants. Whatever $n$ is, the consensus decides at most one single block. As our goal is to scale with the number $n$ of consensus participants, we allow \solution to decide a combination of all the $\Omega(n)$ proposed blocks to make a superblock (line~\ref{line:superblock}). This helps ensuring the \emph{Thousand-TPS} property (\cref{sec:bgrd}) as we explained before.
Note that the same optimization was shown effective for Red Belly Blockchain~\cite{CNG18} to scale to hundreds of consensus participants, but Red Belly only supports the Bitcoin scripting language and not smart contracts.
The drawback of this superblock is that its size increases with 
the number $n$ of participants, and so does its propagation time. 
To cope with arbitrary delays, we build our consensus upon DBFT~\cite{CGLR18} that is partially synchronous~\cite{DLS88} and was recently proved correct via model checking~\cite{TG19,BGK21}.
%
More specifically, this consensus protocol remains safe whatever delay it takes to deliver a message and when messages are delivered in a bounded (but unknown) amount of time the consensus protocol terminates.

\begin{lstlisting}[language=parameterized,basicstyle=\LSTfont,escapechar = ?,escapeinside={(*}{*)},frame = single,firstnumber=20]
index := 0 // consensus instance
blockQueue := (*$\emptyset$*) // pending block to propose
commitChan := chan []byte  // commit channel

start_new_consensus(): (*\label{line:start-new-cons}*)
  index := index + 1 // increment round
  myBlock := blockQueue.peek() // get block proposal
  superblock := consensus_propose(myBlock) // block (*\label{line:propose}*)
  if (myBlock is in superblock) then    (*\label{line:remove-from-queue-start}*)
    blockQueue.poll() // dequeue proposal   (*\label{line:remove-from-queue-end}*)
  commitChan := superblock // send to gRPC srvce (*\label{line:superblock}*)
\end{lstlisting}

\subsubsection{The consensus protocol}
The protocol is divided in two procedures, {\ttt start\_new\_consensus} at lines~\ref{line:start-new-cons}--\ref{line:superblock} that spawns a new instance of (multivalue) consensus by incrementing the replicated state machine {\ttt index}, and {\ttt consensus\_propose} at lines~\ref{line:cons-start}--\ref{line:cons-end} that ensures that the consensus participants find an agreement on a superblock comprising all the proposed blocks that are acceptable. 
The idea of {\ttt consensus\_propose} builds upon classic reduction~\cite{BCG93,BKR94} by executing an all-to-all reliable broadcast~\cite{B87} to exchange $n$ proposals, guaranteeing that any block delivered to a correct process is delivered to all the correct processes: any delivered proposal is stored in an array $\ms{proposals}$ at the index corresponding to the identifier of the broadcaster.
The main difference is that these reductions use a probabilistic binary consensus algorithm while our binary consensus is deterministic.

A binary consensus at index $k$ is started 
with input value {\ttt true} for each index $k$ where a block proposal has been recorded (line~\ref{line:bbc}). 
To limit errors, \solution uses the formally verified deterministic binary consensus of DBFT~\cite{CGLR18}, we omit the pseudocode for the sake of space and
refer the reader to the formal verification of the protocol~\cite{TG19,BGK21}. 
%
%
%
As soon as some of these binary consensus instances return 1, the protocol spawns binary consensus instances with proposed value $\lit{false}$ for each of the non reliably delivered blocks at line~\ref{line:bbc-propose}. Note that this invocation is non-blocking.
As the reliable broadcast fills the {\ttt block} in parallel, it is likely that the blocks reliably broadcast by correct processes have been reliably delivered resulting in as many invocations of the binary consensus with value {\ttt true} instead. Once all the $n$ binary consensus instances have terminated, i.e., {\ttt decidedCount == n} at line~\ref{line:bbc-end}, the superblock is generated with all the reliably delivered blocks for which the corresponding binary consensus returned {\ttt true} (lines~\ref{line:supblock-start}--\ref{line:cons-end}).
At the end of {\ttt start\_new\_consensus}, if the superblock of the consensus contains the block proposed, then this block is removed from the {\ttt blockQueue} at lines~\ref{line:remove-from-queue-start} and~\ref{line:remove-from-queue-end} to avoid reproposing it later.

\begin{lstlisting}[language=parameterized,basicstyle=\LSTfont,escapechar=|,escapeinside={(*}{*)},frame = single,firstnumber=31]
blocks := (*$\emptyset$*) // blocks delivered by reliable bcast

upon reliable_broadcast.deliver(i, block): (*\label{line:rb-deliver}*)
    blocks[i] := block // append block to list
    decBlocks[i] := b_consensus.propose(i, true) (*\label{line:bbc}*)

consensus_propose(myBlock):  (*\label{line:cons-start}*)
  decCount := 0 // # decided bin. cons. instances
  decBlocks :=  (*$\emptyset$*) 
  reliable_broadcast.broadcast(myId, myBlock)  (*\label{line:rbcast}*)
  wait until (*$\exists$*)i : b_consensus.decide(i) == true (*\label{line:bbc-true}*)
     for j  from 0 to n do
       if blocks[j] is null then
         decBlocks[j] := b_consensus.propose(j, false) (*\label{line:bbc-propose}*)
     decCount := decCount+1
   wait until decCount == n (*\label{line:bbc-end}*) 
     superblock := (*$\emptyset$*) (*\label{line:supblock-start}*)
     for i from 0 to n do 
       if decBlocks[i] is true then (*\label{line:decided-blocks}*)
         superblock.add(blocks[i])
  return superblock  (*\label{line:cons-end}*)
 \end{lstlisting}

\subsubsection{Proof-of-stake and membership change}\label{ssec:open}
As mentioned in~\cref{sec:goal}, \solution is an open blockchain and changes its membership (SEVM and consensus nodes) at runtime with a reconfiguring smart contract (provided in Appendix~\ref{appendix}).
To cope with bribery attacks, \solution relies on the proof-of-stake (PoS) design common to other blockchains~\cite{GHM17,Eth2} that assumes that users who have stake are more likely to behave correctly.
Initially, the blockchain is setup with a membership smart contract that accepts a {\ttt rotate} method
that outputs a random sample of $n$ consensus participants among all potential blockchain participants with a preference for participants with the most assets, similar to a sortition~\cite{GHM17}.
Initially and periodically, the correct participants invoke the {\ttt rotate} function that 
outputs new consensus participants for the subsequent blocks.
Traditional SSL authentication 
guarantees that the byzantine participants are ignored. 
As in Eth2~\cite{Eth2}, neither does the system start nor does the {\ttt rotate} method is invoked until sufficiently many participants exist.
To incentivize participants, \solution can reward consensus participants just like Bitcoin's miners~\cite{Nak08}, however, this reward
has not been implemented.

\subsection{Proofs of correctness}\label{line:proof}

In this section, we show that \solution solves the  blockchain problem (Def.~\ref{def:blockchain}). 
Note that the proofs that \solution also guarantees \emph{Lightweigth-Security} and \emph{Thousands-TPS} (\cref{sec:goal}) follows directly from the protocols (\cref{ssec:interface} and \cref{ssec:tx-validation}) and the experimental results (\cref{sec:evaluation}).
For the sake of simplicity in the proofs, we assume that there are as many nodes playing the roles of consensus nodes and state nodes, and one state node and one consensus node are collocated on the same physical machine. 

\begin{lemma}
	\label{theorem1}
	If at least one correct node $\lit{consensus-propose}$s to a consensus instance $i$, then 
	every correct node decides on the same superblock at consensus instance $i$.
\end{lemma}

\begin{proof}
    If a correct node $p$ $\lit{consensus-propose}$s, say $v$, to a consensus instance $i$, then 
    $p$ reliably broadcast $v$ at line~\ref{line:rbcast}.
    By the reliable broadcast properties~\cite{B87}, we know that 
$v$ is delivered at line~\ref{line:rb-deliver} at
all correct nodes.
By assumption, there are at least $2f+1$ correct proposers invoking the reliable broadcast, hence
all correct proposers eventually populate their $\lit{block}$ array with at least one common value. All correct proposers
will thus have input $\lit{true}$ for the corresponding binary consensus instance at line~\ref{line:bbc}.
Now it could be the case that other values are reliably-broadcast by byzantine nodes, 
however, reliable broadcast guarantees that if a correct proposer
delivers a valid value $v$, then all correct proposers deliver $v$. 
By the validity and termination 
properties of the DBFT binary consensus~\cite{CGLR18}, the decided value for the
binary consensus instance  at line~\ref{line:bbc-true} is the same at
all correct nodes. It follows that all correct nodes have the same bit array of $\lit{decBlocks}$ values at line~\ref{line:decided-blocks}
and that they all return the same superblock at line~\ref{line:cons-end} for consensus instance $i$.
%
%
\end{proof}

The next three theorems show that \solution satisfies each of the three properties of the blockchain problem (Definition~\ref{def:blockchain}).
\begin{theorem}
	\solution satisfies the safety property.
\end{theorem}
\begin{proof}
	The proof follows from the fact that any block $B_{\ell}$ at index $\ell$ of the chain is identical for all correct blockchain nodes due to Lemma~\ref{theorem1}.
	Due to network asynchrony, it could be that a correct node $p_1$ is aware of block $B_{\ell+1}$  at index $\ell+1$, whereas another correct node $p_2$ has not created this block $B_{\ell+1}$ yet.
	At this time, $p_2$ maintains a chain of blocks that is a prefix of the chain maintained by $p_1$.
	And more generally, the two chains of blocks maintained locally by two correct blockchain nodes are either identical or one is a prefix of the other. 
\end{proof}

\begin{theorem}
	\solution satisfies the validity property.
\end{theorem}	
\begin{proof}
By examination of the code at line~\ref{line:validate}, only valid transactions are executed and persisted to disk at every correct node.
It follows that for all indices $\ell$, the block $B_{\ell}$ is valid.
\end{proof}

\begin{theorem}\label{thm:liveness}
	\solution satisfies the liveness property.
\end{theorem}	
\begin{proof}
	As long as a correct replica receives a transaction, we know that the transaction is eventually proposed by line~\ref{line:propose}. 
	The proof follows from the termination of the consensus algorithm~\cite{BGK21} and the fact that \solution keeps spawning new consensus instances as long as correct replicas have pending transactions.
\end{proof}

\section{Evaluation of \solution}\label{sec:evaluation}\label{ssec:setup}
In this section, we present the experimental evaluation of \solution, compare it against other blockchains (\cref{sec:comparison}), evaluate it when running across different continents (\cref{sec:geodistributed}) and with a Twitter DApp (\cref{sec:dapp}).

\subsection{Experimental setup}
We use up to 200 AWS virtual machines from 10 regions located in separate countries across 5 continents: Ohio, Mumbai, Seoul, Singapore, Sydney, Tokyo, Canada, Frankfurt, London, Paris, Stockholm, S\~{a}o Paulo. All machines run Ubuntu v18.04.3 LTS, golang v1.13.1. 
When not specified otherwise, the experiments consist of having clients sending 1500 distinct
transactions to each SEVM nodes that exchange with their respective consensus node to spawn a consensus instance. 
The client machines are of type c5.xlarge with 4 vCPUs and 8\,GiB of memory, the SEVM nodes are of type c5.2xlarge 
with 8 vCPUs, 16\,GiB of memory, 
and the consensus nodes are of type c5.4xlarge with 16 vCPUs, 32\,GiB of memory.

As \solution is compatible with Ethereum, we reuse 
libraries of the JavaScript runtime environment {\ttt Node.js}:
we create wallet addresses with {\ttt ethereumjs-wallet}, pre-sign 
transactions to transfer assets, upload or invoke a smart contract with {\ttt ethereumjs-tx}, and serialize these transactions 
before saving them to a JSON file.
The client iterates through the serialized JSON file and sends the transactions to the SEVM 
using {\ttt web3.eth.sendSignedTransaction} through the
{\ttt web3.js} javascript API 
using {\ttt http}. 
Hence, this offloads the encryption time from the performance measurement.
All presented data points are averaged over at least 3 runs.

\subsection{Comparison with other blockchains}\label{sec:comparison}
Here we compare the performance of \solution to Quorum~\cite{jpmorganchase_quorum} from JP Morgan/Consensys and Concord~\cite{concord} from VMware that both support Ethereum smart contracts.
While evaluating all blockchains is out of the scope of this paper, note that Table~\ref{table:comparison} provides a comparison of \solution to non-sharded blockchains while sharded blockchains are discussed in \cref{sec:sharding} and \cref{sec:rw}.
We have also explored Burrow~\cite{burrow} that is unstable~\cite{SNG20} and Ethermint whose open issues ~\cite{EthermintIssue} prevented us from evaluating it.


Figure~\ref{fig:quorum} reports the latency and throughput of both \solution and Quorum.
\solution outperforms Quorum both in terms of latency and throughput because Quorum does not use superblocks.
Interestingly, we also observe that the performance of Quorum starts decreasing as the sending rate increases
whereas the performance of \solution keeps increasing, this is seemingly due to the growing backlog of requests in Quorum 
that induces congestion. Unfortunately, we could not test a sending rate of 800\,TPS and above as Quorum would start losing requests, which confirms previous 
observations~\cite{SNG20}.


\begin{figure}[t]
	\centering
	\includegraphics[scale=0.28,clip=true,viewport=55 0 960 290]{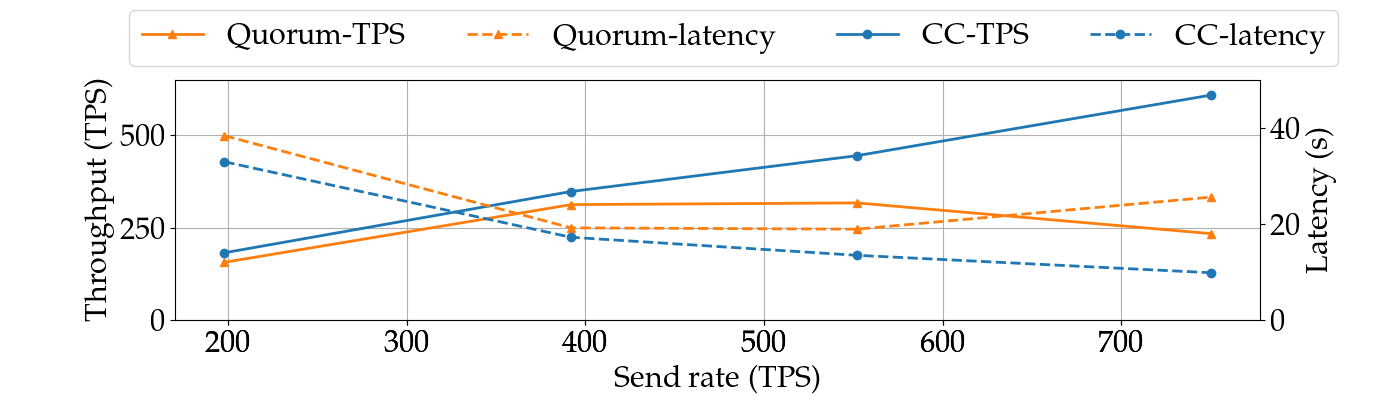}
	\caption{Comparison of throughput and latency between \solution (CC) and Quorum against sending rate. 
	}
	\label{fig:quorum}
\end{figure}

\begin{figure}[t]
	\centering
	\includegraphics[scale=0.28,clip=true,viewport=55 0 960 290]{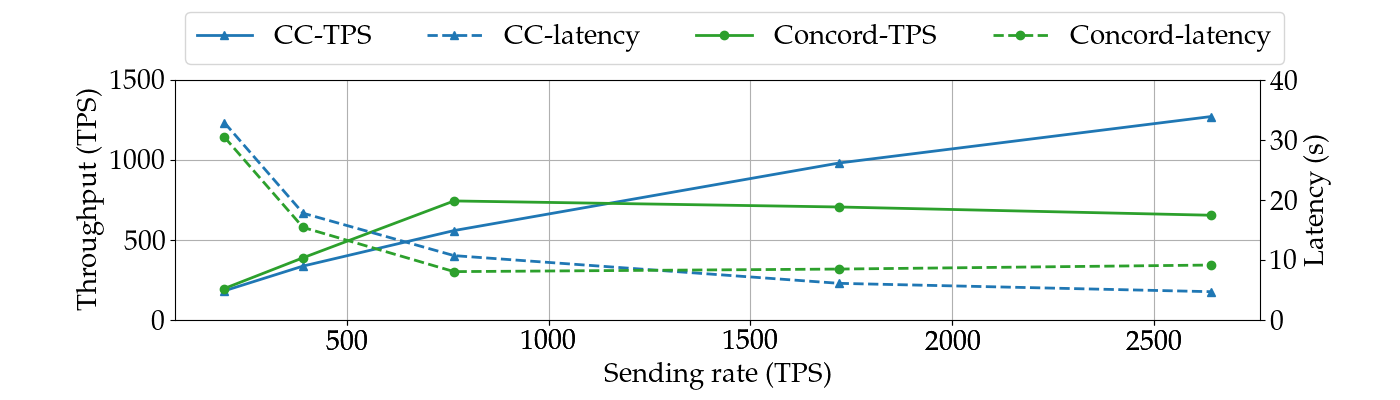}
	\caption{Comparison of throughput and latency between \solution (CC) and Concord against sending rate.}
	\label{fig:concord}
\end{figure}

%
%
Figure~\ref{fig:concord} compares the throughput and latency of Concord and \solution. 
As Concord suffers from known configuration issues~\cite{bugconcord20} that prevented us from running it on a distributed system, we ran
both Concord and \solution on a single c5.9xlarge machine with 4 client machines. 
Concord slightly outperforms \solution with low sending rates, however, as the sending rate increases, \solution outperforms Concord significantly.

\subsection{Effect of validation reduction using the SmartBank DApp}\label{sec:dapps}
In order to assess the impact of the validation optimization (\cref{sec:sevm}) of the SEVM on the performance, 
we measured the time spent validating eagerly when running the smart bank DApp that is part of BlockBench~\cite{DWC17}. To this end, 
we instrumented its {\ttt writeCheck} function to measure both the total time $\Delta^{n}_{SEVM}$ spent treating $k$ calls and the average time 
$\delta^{n}_{SEVM}$ spent by each server of SEVM validating eagerly these calls on $n$ nodes, 
to deduce the rest of the treatment time not affected by the validation optimization $\beta = \Delta^{n}_{SEVM} - \delta^{n}_{SEVM}$.

Based on this measurement, we could deduce the time $\delta_{EVM}$ the EVM would spend 
validating eagerly without the validation optimization:
$\delta_{EVM} = n\cdot \delta^{n}_{SEVM}.$
In particular, regardless of $n$, we know that the EVM would spend $\Delta_{EVM} = \beta + \delta_{EVM}$ to treat the function calls.
By contrast, depending on $n$, the SEVM would spend $\Delta^{n}_{SEVM} = \beta + \delta^{n}_{SEVM}$.
As $\delta^{n}_{SEVM} =  \delta_{EVM}/n$, we know that $\lim_{n\rightarrow \infty}(\delta^{n}_{SEVM}) = 0$.
This means that, with $n$ servers, 
the EVM slowdown compared to the SEVM is:
$$S = \frac{\Delta_{EVM} - \Delta^{n}_{SEVM}}{\Delta^{n}_{SEVM}} = \frac{\delta_{EVM}+\delta^{n}_{SEVM}}{\beta+ \delta^{n}_{SEVM}}.$$
As $n$ tends to infinity, we thus have a slowdown of: $$\lim_{n\rightarrow +\infty}S = \frac{\delta_{EVM}}{\beta}.$$

Our measurement obtained with $k=6000$ transactions and $n=4$ revealed that 
$\delta^{n}_{SEVM} = 0.61$ seconds and $\Delta^{n}_{SEVM} = 5.66$ seconds.
Hence, we have $\beta =  5.66-0.61 = 5.05$.
As $n=4$, we have $\delta_{EVM} = 4 \times 0.61 = 2.44$ so that $\Delta_{EVM} = 5.05+2.44 = 7.49$.
This means that the EVM would take $S = 32\%$ more time than the SEVM to treat these DApp requests.
Finally, as $n$ tends to infinity, the slowdown of the EVM over the SEVM would become $48\%$.

\subsection{Storing the superblock efficiently}\label{ssec:store}
To measure the impact on performance of storing each block separately,
we implemented a naive method that stores the whole block at once.
More precisely, we changed the storing loop (\cref{ssec:storing}) 
by simply removing the inner {\ttt for} loop at 
lines~\ref{line:block-loop-start}--\ref{line:block-loop-end} that persisted one (sub-)block at a time, in order 
to persist the superblock once and for all. 
In this experiment, we setup a network of 2 clients machines, 10 consensus machines and 10 SEVM machines where clients send 1500 distinct transactions to each EVM node for a total of 15000 transactions.
\remove{
This results in the loop depicted below:

\begin{lstlisting}[language=parameterized,basicstyle=\LSTfont,escapechar = ?,escapeinside={(*}{*)},frame = single]
execute_transaction:
  // for each superblock received 
  for block := range node.commitChan do
    total := 0
    // deserialize the block
    txs := node.txm.deserialize(block)
    // create the corresponding block structure
    updateBlockState(&tmtAbciTypes.Header{
         Time: time.Now(), NumTxs: int64(len(txs))})
    // execute each transaction in the block
    for _, tx := range txs do // for each txn
      executeTx(tx)
    done
    // persist block by writing to disk
    persist(common.Address{})
    total += len(txs)
  done
\end{lstlisting}
}
\begin{figure}[ht]
	\centering
	\includegraphics[scale=0.27,clip=true,viewport=20 0 930 355]{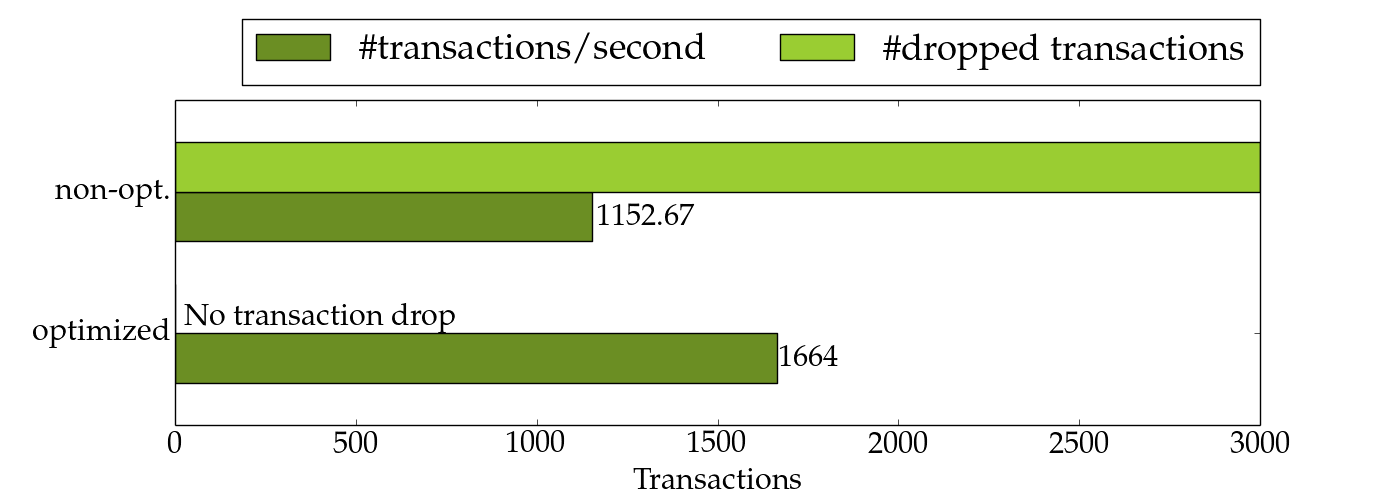}
	\caption{Performance difference when processing each block of a superblock at a time (optimized) and when processing the entire superblock at once (non-optimized)}
	\label{fig:superblock}
\end{figure}

Figure~\ref{fig:superblock} compares the performance obtained with \solution (superblock optimized) and with the naive approach (superblock non-optimized). 
The throughput of \solution (superblock optimized) is $44\%$ higher than the throughput of the naive approach.
This is because trying to persist a large superblock that comprises $10$ blocks 
leads to I/O congestion. 
One might argue that multi-threading block writes and transaction executions could also solve this issue. However, this is not possible as the execution and writing of blocks should happen sequentially. Interestingly, 
In addition, we observed that 3000 transactions get dropped, which represents 20\% of all transactions, when executing the naive approach.
This is due to CPU overload: executing a superblock of 10 blocks within a single loop iteration is more CPU intensive than executing one block per iteration 
because between two block executions the CPU resource can be allocated to other tasks.
If the clients keep sending transactions while the CPU usage of the SEVM node reaches 100\%, the SEVM starts dropping incoming transactions as soon as it cannot hold any more transactions.
These results show the importance of optimizing the superblock storage for \solution to not suffer transaction drops.

\subsection{World-wide scalability}\label{sec:geodistributed}

To evaluate the scalability of the performance of \solution, we deployed \solution in 10 regions spanning 5 continents: Canada, London, Mumbai, Oregon, Paris, S\~{a}o Paulo, Singapore, Stockholm, Sydney and Tokyo. 
As previously mentioned, we consider for simplicity that each participant is running both a consensus node and an SEVM node so that we can consider each participant as a single entity, out of all of which at most a third can be faulty.

Figure~\ref{fig:geo} depicts the throughput without end-to-end ecryption (w/o TLS) and with encryption (with TLS) of \solution as we run \solution on more and more machines: We start our experiment with 20 machines spread evenly in the 10 countries and add machines by group of 20 evenly spread in the 10 countries
until we reach 200 machines.
We observe that the throughput increases as we increase the number of nodes from 1100\,TPS at 20 machines to 2038\,TPS at 200 machines, demonstrating the scalability of \solution even in a geo-distributed setting.
The curve flattens out at large scale between 140 and 200 nodes, indicating that the gain obtained in throughput by adding more machine becomes lower and lower. This is due to the numerous machines consuming the available bandwidth.
Finally, we observe, as expected, that the TLS encryption comes at a cost, however, this overhead is negligible in comparison of the overall performance as the peak throughput with TLS (1960\,TPS) is only 4\% lower than the peak throughput without TLS (2038\,TPS).


\begin{figure}[th]
	\centering
	\includegraphics[scale=0.32,clip=true,viewport=25 0 1000 410]{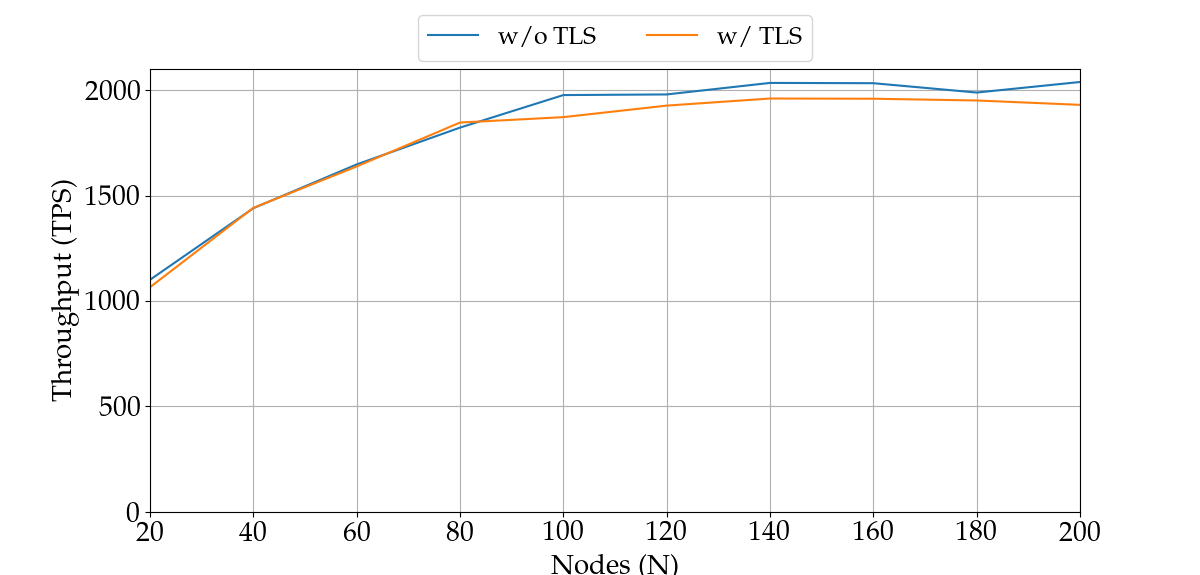}
	\caption{Throughput of \solution when deployed in a geo-distributed AWS environment of 10 countries across 5 continents with and without TLS\label{fig:geo}}
\end{figure}

Figure~\ref{fig:latency-geo} shows the latency of transactions of \solution in the aforementioned geo-distributed environment as the number of nodes increases. 
We can observe that the latency increases with the number of nodes.
We observe similar minimum latencies across all system sizes but the $99^{th}$ percentile indicates that some requests can take much longer especially at large scale: the transactions take less than 10 seconds to execute on up to 40 nodes while they take less than 40 seconds to execute at 200 nodes. 
It is important to note that these latencies can be viewed as time for a transaction to become final: thanks to our deterministic byzantine fault tolerance consensus (\cref{sec:dbft}), transactions are committed (and thus final) as soon as the consensus ends and the superblock is selected. This differs from classic blockchains~\cite{yakovenko2018solana,HMW} whose consensus is reached after the block is appended and after more ``block confirmations'' occur.
Interestingly, this increasing latencies do not prevent the throughput from scaling with the number of machines as we discussed earlier (Fig.~\ref{fig:geo}).
This is precisely due to the superblock optimization: As more machines participate, more blocks get proposed and running consensus takes more time, which increases the latency, however, the number of transactions decided per consensus instance also increases, which guarantees scalability.

\begin{figure}[th]
	\centering
	\includegraphics[scale=0.38,clip=true,viewport=25 0 1000 410]{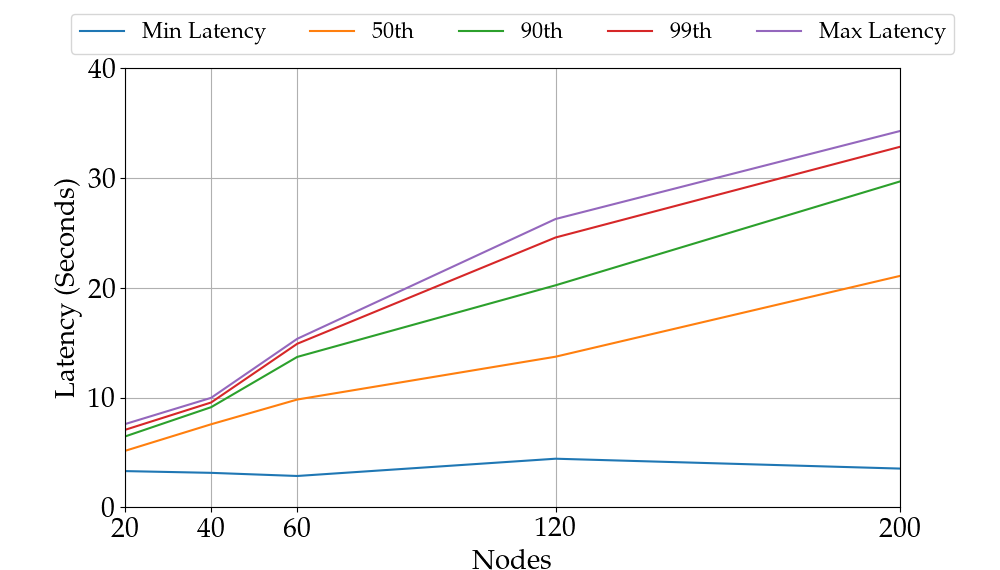}
	\caption{Latency of \solution when deployed in a geodistributed AWS environment of 10 countries across 5 continents\label{fig:latency-geo}}
\end{figure}

\subsection{Twitter DApp evaluation}\label{sec:dapp}

To evaluate how fast \solution can treat smart contract invocations under a realistic workloads, 
we ran the Twitter DApp of the {\sc Diablo} framework~\cite{BGG21} on top of 4 consensus nodes and 4 SEVM nodes and 
report on the performance as time elaspses. {\sc Diablo} is a benchmark suite for blockchains that features DApps written in different smart contract programming languages.
It features a Twitter DApp written in Solidity whose smart contract sends 140-character messages following a burst workload experienced during the release of the \emph{Castle in the Sky} anime. 

\begin{figure}[ht]
	\centering
	\includegraphics[scale=0.3]{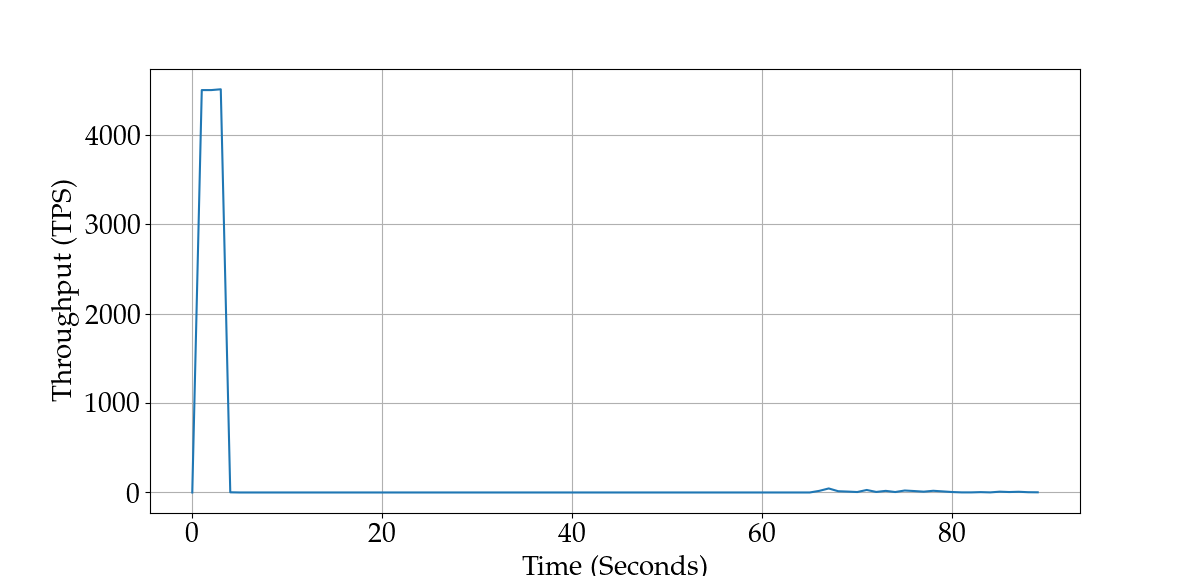}
	\caption{Twitter DApp throughput when running on top of \solution\label{fig:twitter}
	}
\end{figure}

Figure~\ref{fig:twitter} depicts the performance results obtained while running this Twitter DApp on top of \solution.
To achieve the high burst Twitter workload of 143,000\,TPS, we had to deploy as many clients as SEVM nodes.
We can observe that \solution handles the load burst of the Twitter workload by committing up to 4509 transactions per second.
During the same experiments, we observed that a latency with $50^{th}$, $90^{th}$, $99^{th}$ percentiles as 5.367, 8.171 and 195.337 seconds.
Despite this high workload burst, \solution continues executing properly, committing other transactions at a slower rate with a minor peak around 67 seconds.
This demonstrates that \solution can deliver high throughput and handles workload bursts as can be found in realistic applications.


\subsection{Linear speedup with sharding}\label{sec:sharding}

As we explain in the related work, sharding proved instrumental in boosting performance 
of blockchains. The idea of sharding stems from distributed database research where
a database table gets split into sub-tables, each sub-table is stored on a partition of machines called a \emph{shard}.
When a request on a entry of the table is issued, the shard managing this entry is responsible for handling this request. 
Hence, requests on distinct entries can execute in parallel on distinct shards, allowing performance to increase (ideally linearly) 
with the number of shards.

In blockchain, there exists various ways of implementing sharding.
With traditional blockchain sharding, each shard is responsible of a subset of the transactions
and runs an independent consensus instance to agree upon the ordering of these transactions.
This is the approach taken by Dfinity~\cite{HMW}, Elastico~\cite{LNZ16}, RapidChain~\cite{ZMR18} and Omniledger~\cite{KJGG18}.
With verification sharding, every shard participates in the same consensus instance and stores the 
global state, however, each shard is responsible of verifying a different set of transactions.
This is the approach taken by Red Belly Blockchain~\cite{CNG21}.

We implemented sharding on top of \solution using a traditional blockchain sharding. 
Shards can be spawned on demand by a dedicated built-in smart contract, hence resulting
in a beacon chain with shard chains structure similar to the upcoming Ethereum 2~\cite{Eth2}.
Participants can deposit some assets on the beacon chain in order to spawn a new shard, which creates
a new blockchain instance where the participants have an account with a balance corresponding to the assets they 
deposited on the beacon chain---this typically allows participants to transact within the shard without having to 
transact on the beacon chain.

\begin{table}[ht]
    \centering
    \setlength{\tabcolsep}{12pt}
    \begin{tabular}{r|lll}
        \toprule
         \#Shards &  1  & 2  & 3  \\
         \midrule
         Throughput & 470.79\,TPS & 951.54\,TPS & 1405.42\,TPS \\
         \bottomrule
    \end{tabular}
    \vspace{0.5em}
    \caption{\solution performance scales almost linearly with the number of shards\label{fig:shard}}
    \label{table:sharding}
\end{table}

Table~\ref{table:sharding} presents the performance obtained on up to 32 small machines with 2 vCPU and 8\,GiB memory running Ubuntu 20.04
when we increase the number of shards. The beacon chain (1 shard) delivers 470.79\,TPS but when coupled with two other shards (3 shards)
it delivers $2.98\times$ higher performance, which demonstrates a speedup very close to linear. This is no suprise given 
that each shard runs on separate machines and use distinct resources. Although sharding makes the implementation of cross-shard 
transactions quite complex, 
we could offer cross-shard transactions by decoupling these transactions into separate withdrawals an credits, as was previously 
suggested in Prism~\cite{WWB20}, for applications where the atomicity of transactions is not a requirement.

\section{Related Work}\label{sec:rw}

To decentralize the computation from large data stores~\cite{DHJ07,CDG08,CDE13}, various work focused on user/edge-centric computing~\cite{GME15}. Solid~\cite{MSH16} distributes private data into pods whose user manages permissions. Lightweight middleware~\cite{JLJ19} exploit WebRTC
to avoid downloading a blockchain. These solutions do not offer the execution transparency  of blockchains, which is key to prevent user manipulations~\cite{Bria20}.



\paragraph{Payment blockchains}
Some blockchains are designed for high transactions throughput at large scale, but were not designed to support DApps~\cite{GHM17,CNG18,SDV19,LLM19,GRH20}. This is the case of ResilientDB~\cite{GRH20} that exploits topology-awareness to parallelize consensus executions, the
Red Belly Blockchain~\cite{CNG18} that shares our superblock optimization or Mir~\cite{SDV19} that 
 deduplicates transaction verifications.
Stellar~\cite{LLM19} is an in-production blockchain running in a geodistributed setting while
Algorand~\cite{GHM17} introduced the sortition our membership change builds upon. 
Although progresses are being made towards smart contract support, these blockchains do not run DApps.


\paragraph{Fast smart contract executions}
Solana~\cite{yakovenko2018solana} builds upon Proof-of-History (PoH) to reduce message overhead. 
Solana provides high performance by offering optimistic consensus, hoping that a single block gets notarized at each index, and
thanks to the vertical scaling of its validator nodes: validator nodes feature 1TB SSD disk and 2 Nvidia V100 GPUs for benchmarking~\cite{SolanaPerf} and 128 GiB memory is required~\cite{SolanaValidReq}, which is twice as much as the most powerful machines we used
in our experiments. 

\paragraph{Towards byzantine fault tolerant blockchains}
Upper-bounding the number $f$ of byzantine failures allow to solve consensus to avoid forks.
Ethereum comes with proof-of-work and 
proof-of-authority (PoA) in the two mainstream Ethereum programs, called {\ttt parity} and {\ttt geth}. The idea of proof-of-authority is to have a set of $n$ permissioned validators, among which $f$ can be malicious or \emph{byzantine}~\cite{PSL80}, that generate new blocks~\cite{noauthor_ethereum_nodate,noauthor_poacore_nodate,BBK18}.
Unfortunately, 
both proof-of-authority protocols in {\ttt parity} and {\ttt geth} have recently been shown vulnerable to the attack of the clone when messages take longer than expected~\cite{EGJ20}. 

\paragraph{Tolerating unpredictable bounded delays}
To cope with unpredictable message delays, blockchains cannot rely on synchrony.
%
%
Cosmos~\cite{CFM21}, sometimes referred to as the Internet of Blockchain, is a network of interoperable blockchains that builds upon the Tendermint state machine replication~\cite{Buc16}.
Ethermint~\cite{Ethermint} is a blockchain that combines the partially synchronous Tendermint consensus protocol~\cite{BKM18} with the EVM. 
Ethermint is still under active development~\cite{EthermintPage} and we could not benchmark it.
In particular, we found some issues that prevented us from deploying it
like 
a nonce management limitation, which resulted in rejecting consecutive transactions sent in a short period of time~\cite{EthermintIssue}. Other researchers who managed to deploy an older version of Ethermint, reported a peak throughput of 100\,TPS obtained with a single validator node~\cite{DBD18}, however, Tendermint reached 438\,TPS~\cite{CFM21}.

Zilliqa~\cite{ZILLIQA} is a blockchain that supports smart contracts and reaches consensus with PBFT~\cite{CL02}. 
We are not aware of any performance evaluation of Zilliqa but its  state machine, Scilla, executes non Turing complete programs but slower than the EVM when the state size increases~\cite{Scilla}.
Therefore, it is unlikely that it would yield higher throughputs than our \solution for large state sizes.
Chainspace~\cite{albassam2017chainspace} introduced a distributed atomic commit protocol termed S-BAC for  smart contract transactions. Coupled with the BFT-SMaRt~\cite{BSA14} consensus protocol, Chainspace can support trustless use of DApps. However, 
it has only been able to achieve up to 350\,TPS, offering a limited support for DApps.

\paragraph{Evaluations of BFT blockchains}
Quorum~\cite{jpmorganchase_quorum} is a blockchain that supports Ethereum smart contracts and reaches consensus with the Istanbul Byzantine Fault Tolerant (IBFT) consensus algorithm. 
Just like \solution, the byzantine fault tolerance of Quorum makes it well-suited for mobile devices to interact wth DApps securely without downloading the blockchain. Moreover, it seems that few optimizations could help it treat a large number of transactions per second~\cite{BSK18}. Unfortunately, Quorum
loses requests (\cref{sec:comparison}).

SBFT~\cite{GAG19} is a byzantine fault tolerant consensus algorithm that exploits threshold signatures to reduce the communication complexity of PBFT but commits, like PBFT, at most one proposed 
block per consensus instance. It was shown to reach consensus on 378 smart contract requests per second when deployed within one continent and 172 requests per second across multiple continents. 
Concord~\cite{concord} is a blockchain that combines a lightweight C++ implementation of the EVM with SBFT, however, its publicly available version has open issues~\cite{bugconcord20} that prevent it from being deployed on distinct physical machines but we showed that Concord, although slower than \solution, reached the encouraging throughput of 1000\,TPS on 4 nodes within the same physical machine. It could be the case that future versions will scale. 

\paragraph{Sharding}
As we presented in \cref{sec:sharding}, one can multiply the performance of a blockchain, including \solution, by adding more shards.
%
Dfinity~\cite{HMW} coined as the Internet Computer is an open permissioned blockchain.
Dfinity scales horizontally thanks to its committees, with an assumed majority of correct members, that act like shards. It achieves high block production throughput thanks to concurrent execution of cannisters, isolated pieces of code compiled to WASM that act as smart contracts and offer low latency to read requests. 
The difference with \solution, is that a block produced is not necessarily final: a verifiable random function is used 
to rank block proposers and if an adversarial one is rank highest, it could propose conflicting blocks that are notarized, hence leading to a fork.
Additional assumptions are needed for the nodes to agree on the chain with the highest block weight.
\solution solves consensus before appending blocks.

The {\ttt move} approach~\cite{FBP20} moves accounts and computation from one smart contract enabled blockchain to another. The smart contract of the first blockchain is locked before any participant creates it in the second blockchain. This allows to scale the throughput of the congested DApp CryptoKitties with the number of shards.
Eth2~\cite{Eth2} relies on a beacon chain and will feature 64 shard chains to improve the scalability of Ethereum. The uniqueness of the beacon chain guarantees a consistent view of current state but cannot handle accounts and smart contracts. The validators of the shard chain do not need to download and 
run data for the entire network.

{\sc Prism}~\cite{PRISM} is a proof-of-work blockchain that shards the blockchain into $m$ voter chains and exploits three types of blocks in a block tree.
The voter blocks are used to vote for proposer blocks grouped per level in the block tree.
Once a proposer block is elected, transaction blocks that are pointed to by the proposer block are committed.
Prism peaks at 19K\,TPS by 
ignoring the eager validation completely,
which exposes it to DoS attacks.
To ensure the copy of the blockchain state is not corrupted, a user needs first to download the block headers, a time- and space-consuming task ill-suited for running DApps on handheld devices.


\section{Conclusion}\label{sec:conclusion}

\solution is a collaborative blockchain compatible with the largest ecosystem of DApps that treats thousands of requests per second and scales to hundreds of machines world-wide. It builds upon recent advances in deterministic byzantine fault tolerance (BFT) consensus algorithms to avoid forks and offers finality without having to wait for block confirmations.
Its key novelties lie in (i)~having smart contract execution nodes collaborating to minimize validations and in (ii)~having
consensus nodes collaborating to combine their block proposals into a committed superblock of smart contracts.

Our experiments demonstrate that \solution is an appealing BFT middleware for individuals to exchange in a fully distributed fashion.
We showed that one instance of \solution handles a peak throughput of 4500\,TPS under a Twitter DApp. 
We also showed how to interconnect different shard instances of \solution to scale almost linearly, to support 
potentially as many DApps as shards. This,
combined with the ability of one instance (or shard) to scale to hundreds of nodes spread in 5 continents makes \solution an appealing BFT middleware for DApp services. 
This new model departs from the centralization trend of the sharing economy services to offer more transparent and fault tolerant services to individuals.


\bibliographystyle{plain}
\bibliography{reference.bib} 

\appendix
\subsection{Smart Contract to Reconfigure Nodes}\label{appendix}

\onecolumn
{\tiny
\begin{lstlisting}[language=Solidity,basicstyle=\scriptsize\LSTfont,escapechar = ?]
pragma solidity ^0.4.0;
pragma experimental ABIEncoderV2;

contract Committee {
    
     string [] committee;
     address public chairperson;
     uint member;
     mapping(uint => uint) private id;
     mapping(string => bool) hasIp;
     mapping(string => bool) hasCalled;
     mapping(address => string) WallettoIP;
     constructor() public {
       chairperson = msg.sender;
       for (uint d=0; d < 10; d++){
          id[d] = 0;
       }
     }
     uint [] count;
     string [] public selected;
     uint32 [] private digits;
     uint threshold;
     uint size;
     uint select;
     uint option;
     event notify(string []);
     
     // initial set of node ips and the size of the committee is parse by the chairperson
     function addIp (string [] memory ip, uint members, string[] memory wallets) public {
        committee.length = 0;
         require(
            msg.sender == chairperson,
            "Only chairperson can give right to vote."
        );
        for (uint t = 0; t < ip.length; t++){
           committee.push(ip[t]);
           WallettoIP[parseAddr(wallets[t])] = ip[t];
           hasIp[ip[t]] = true;
           hasCalled[ip[t]] = false;
        }
        size = committee.length;
        member = members;
        // members is the number of participants per committee
     }
     
     // upon parsing a random seed and once a threshold of calls have been made, a random committee is selected
     function createCommittee (uint val) public {
         // if the caller of this function is in the list of ips added by the chairperson, and if they haven't call this function
         // before - because we don't want t+1 be reached by a malicious node calling this function multiple times
         require(hasIp[WallettoIP[msg.sender]] == true && hasCalled[WallettoIP[msg.sender]] == false);
         hasCalled[WallettoIP[msg.sender]] = true;
         id[val] = id[val] + 1;
         // use size instead of member
         threshold = (size - 1)/3 + 1;
         if (id[val] == threshold){
           select = size/member; // total number of nodes divided by the ones that should be in the committee. 
           // gives the possible number of different committees
           for(uint b =0; b<select;b++){
               count.push(b);
           }
           // get the number of options to the range of select
           option = val % select;
           for (uint a = 0; a<count.length; a++){
            if(option==count[a]){
                for (uint j=count[a]*member; j < (count[a] + 1)*member; j++){
                    selected.push(committee[j]);
                }   
            }
           }
           emit notify(selected);
           selected.length = 0;
           count.length = 0;
           option = 0;
           for (uint d=0; d < 10; d++){
                id[d] = 0;
            }
            for (uint num=0; num < size; num++){
                hasCalled[committee[num]] = false;
            }
        }
            
     }
     
    // a function to convert addresses of type string to type address     
    function parseAddr(string memory _a) public returns (address _parsedAddress) {
        bytes memory tmp = bytes(_a);
        uint160 iaddr = 0;
        uint160 b1;
        uint160 b2;
        for (uint i = 2; i < 2 + 2 * 20; i += 2) {
            iaddr *= 256;
            b1 = uint160(uint8(tmp[i]));
            b2 = uint160(uint8(tmp[i + 1]));
            if ((b1 >= 97) && (b1 <= 102)) {
                b1 -= 87;
            } else if ((b1 >= 65) && (b1 <= 70)) {
                b1 -= 55;
            } else if ((b1 >= 48) && (b1 <= 57)) {
                b1 -= 48;
            }
            if ((b2 >= 97) && (b2 <= 102)) {
                b2 -= 87;
            } else if ((b2 >= 65) && (b2 <= 70)) {
                b2 -= 55;
            } else if ((b2 >= 48) && (b2 <= 57)) {
                b2 -= 48;
            }
            iaddr += (b1 * 16 + b2);
        }
        return address(iaddr);
    }  
}
\end{lstlisting}
}

\end{document}